\documentclass[aps,pra,twocolumn,superscriptaddress,groupedaddress,nofootinbib]{revtex4}  % for review and submission
\usepackage{graphicx}  % needed for figures
\usepackage{dcolumn}   % needed for some tables
\usepackage{bm}        % for math
\usepackage{amssymb}   % for math

\usepackage{hhline}
\usepackage{scalefnt}
\usepackage{amsmath}
\usepackage{color}
\usepackage{xr}
\usepackage{cases}
\usepackage{enumerate}
\usepackage{tikz}
\usetikzlibrary{matrix,arrows}
\usepackage{enumitem}
 \usepackage{relsize}
\def\I {\mathrm{i}}

\usepackage{slashbox}
\usepackage{xy}
\xyoption{matrix}
\xyoption{frame}
\xyoption{arrow}
\xyoption{arc}

\usepackage{ifpdf}
\ifpdf
\else
\PackageWarningNoLine{Qcircuit}{Qcircuit is loading in Postscript mode.  The Xy-pic options ps and dvips will be loaded.  If you wish to use other Postscript drivers for Xy-pic, you must modify the code in Qcircuit.tex}
\xyoption{ps}
\xyoption{dvips}
\fi

\entrymodifiers={!C\entrybox}

\newcommand{\qw}[1][-1]{\ar @{-} [0,#1]}
\newcommand{\qwd}[1][-1]{\ar @{.} [0,#1]} %dashed version of \qw
\newcommand{\qwx}[1][-1]{\ar @{-} [#1,0]}
\newcommand{\cw}[1][-1]{\ar @{=} [0,#1]}
\newcommand{\cwx}[1][-1]{\ar @{=} [#1,0]}
\newcommand{\gate}[1]{*+<.6em>{#1} \POS ="i","i"+UR;"i"+UL **\dir{-};"i"+DL **\dir{-};"i"+DR **\dir{-};"i"+UR **\dir{-},"i" \qw}   
    \newcommand{\ucc}[1]{*+<.6em>{#1} \POS ="i","i"+UR;"i"+UL **\dir{-};"i"+DL **\dir{-};"i"+DR **\dir{-};"i"+UR **\dir{-},"i" \cw}
\newcommand{\meter}{*=<1.8em,1.4em>{\xy ="j","j"-<.778em,.322em>;{"j"+<.778em,-.322em> \ellipse ur,_{}},"j"-<0em,.4em>;p+<.5em,.9em> **\dir{-},"j"+<2.2em,2.2em>*{},"j"-<2.2em,2.2em>*{} \endxy} \POS ="i","i"+UR;"i"+UL **\dir{-};"i"+DL **\dir{-};"i"+DR **\dir{-};"i"+UR **\dir{-},"i" \qw}
\newcommand{\control}{*!<0em,.025em>-=-<.2em>{\bullet}}
\newcommand{\ctrl}[1]{\control \qwx[#1] \qw}
\newcommand{\targ}{*+<.02em,.02em>{\xy ="i","i"-<.39em,0em>;"i"+<.39em,0em> **\dir{-}, "i"-<0em,.39em>;"i"+<0em,.39em> **\dir{-},"i"*\xycircle<.4em>{} \endxy} \qw}
\newcommand{\ctarg}{*+<.02em,.02em>{\xy ="i","i"-<.39em,0em>;"i"+<.39em,0em> **\dir{-}, "i"-<0em,.39em>;"i"+<0em,.39em> **\dir{-},"i"*\xycircle<.4em>{} \endxy} \cw}
\newcommand{\multigate}[2]{*+<1em,.9em>{\hphantom{#2}} \POS [0,0]="i",[0,0].[#1,0]="e",!C *{#2},"e"+UR;"e"+UL **\dir{-};"e"+DL **\dir{-};"e"+DR **\dir{-};"e"+UR **\dir{-},"i" \qw}
\newcommand{\multigateUp}[3]{*+<1em,.9em>{\hphantom{#2}} \POS [0,0]="i",[0,0].[#3,0] !C *{#2}, \POS [0,0]="i",[0,0].[#1,0]="e", "e"+UR;"e"+UL **\dir{-};"e"+DL **\dir{-};"e"+DR **\dir{-};"e"+UR **\dir{-},"i" \qw} %Reposition the gate name by giving the position (natural number) as a third argument (the bigger the number, the lower the name will appear)
\newcommand{\ghost}[1]{*+<1em,.9em>{\hphantom{#1}} \qw}
\newcommand{\lstick}[1]{*!R!<.5em,0em>=<0em>{#1}}
\newcommand{\Qcircuit}{\xymatrix @*=<0em>}
\newcommand{\ket}[1]{\left| #1 \right>} % for Dirac bras
\newcommand{\ketbra}[2]{|#1\rangle\!\langle#2|}

\newcommand{\id}{\leavevmode\hbox{\small1\normalsize\kern-.33em1}}
\newcommand{\tr}{\mathrm{tr}}
\newcommand{\nt}{{\sc not}}
\newcommand{\cnot}{{\sc C-not}}
\newcommand{\cnots}{\cnot s}

\newcommand\restr[2]{{% we make the whole thing an ordinary symbol
  \left.\kern-\nulldelimiterspace % automatically resize the bar with \right
  #1 % the function
  \vphantom{\big|} % pretend it's a little taller at normal size
  \right|_{#2} % this is the delimiter
  }}

\newtheorem{thm}{Theorem}
\newtheorem{lem}[thm]{Lemma}
\newtheorem{cor}[thm]{Corollary}
\newtheorem{rmk}{Remark}

\newenvironment{proof}[1][Proof]{\noindent\textbf{#1.} }{\ \rule{0.5em}{0.5em}}

\begin{document}

\title{Quantum Circuits for Quantum Channels}

\author{Raban~Iten} \email{itenr@itp.phys.ethz.ch} \affiliation{Institute for Theoretical Physics, ETH Z\"urich, 8093 Z\"urich, Switzerland}  
\author{Roger~Colbeck} \email{roger.colbeck@york.ac.uk} \affiliation{Department of Mathematics, University of York, YO10 5DD, UK}               
\author{Matthias~Christandl} \email{christandl@math.ku.dk}
\affiliation{QMATH, Department of Mathematical Sciences, University of
  Copenhagen, Universitetsparken 5, 2100 Copenhagen, Denmark}    

\date{$10^{\mathrm{th}}$ May 2017}

\begin{abstract}
  We study the implementation of quantum channels with quantum
  computers while minimizing the experimental cost, measured in terms
  of the number of Controlled-{\nt} (\cnot{}) gates required (single-qubit gates are free). We consider three different models. In the
  first, the Quantum Circuit Model (QCM), we consider sequences of
  single-qubit and \cnot{} gates and allow qubits to be traced out at
  the end of the gate sequence. In the second (RandomQCM), we also
  allow external classical randomness. In the third (MeasuredQCM) we
  also allow measurements followed by operations that are classically
  controlled on the outcomes. We prove lower bounds on the number of
  \cnot{} gates required and give near-optimal decompositions in
  almost all cases. Our main result is a MeasuredQCM circuit for any
  channel from $m$ qubits to $n$ qubits that uses at most one ancilla
  and has a low \cnot{} count. We give explicit examples for small
  numbers of qubits that provide the lowest known \cnot{} counts.
\end{abstract}

\maketitle

\section{Introduction}
Quantum channels, mathematically described by completely positive,
trace-preserving maps, play an important role in quantum information
theory because they are the most general evolutions quantum systems
can undergo. The ability to experimentally perform an arbitrary
channel enables the simulation of noisy channels. For example,
this is useful to test how a new component (e.g., a receiver) will
perform when subjected to noise in a more controlled
environment. Moreover, experimental groups can show their command over
quantum systems consisting of a small number of qubits by
demonstrating the ability to perform arbitrary quantum channels on
them (see for example~\cite{open_quantum_simulation} and references
therein).  Instead of building a different device for the
implementation of each quantum channel, it is convenient to decompose
arbitrary channels into a sequence of simple-to-perform operations. In
this paper we work with a gate set consisting of \cnot{} and
single-qubit gates, which is universal~\cite{5}. However, note that the main ideas of our circuit constructions and lower bounds generalize to arbitrary gate sets. The implementation of
a \cnot{} gate is usually more prone to errors than the implementation
of single-qubit gates. For example, the lowest achieved infidelities
are by a factor of more than 10 smaller for single-qubit gates than
for two qubit gates~\cite{Ballance_new, Gaebler}. This motivates using
the number of \cnot{} gates to measure the cost of a quantum circuit.\footnote{Certain experimental architectures only allow nearest neighbour \cnot{} gates. Here we assign the same cost to each \cnot{} gate regardless of its interaction distance. However, our circuit constructions for channels are based on decompositions of isometries, which are straightfoward to adapt to the nearest neighbour case~\cite{Iso,2}.}

In this work we consider the construction of universal circuit
topologies comprising gates from this universal set.  A circuit
topology~\cite{unitary_lowerb1, unitary_lowerb2} corresponds to a set
of quantum channels that have a particular structure but in which some
gates may be free or have free parameters. Our aim is to find circuit
topologies that minimize the \cnot{} count but are universal in the
sense that any channel from $m$ to $n$ qubits can be obtained by
choosing the free parameters appropriately.

We work with three different models. In the first we
consider \emph{the quantum circuit model} (QCM), in which we allow a
sequence of \cnot{}, single-qubit gates and partial trace operations
on the qubits and any ancilla.  In the second (RandomQCM) we allow the
use of classical randomness in addition.  In the third (MeasuredQCM),
we allow the operations of the QCM as well as measurements and
operations that are classically controlled on the measurement
outcomes.

A task that is related to the construction of a circuit topology is
that of minimizing the \cnot{} count for a given quantum channel (on a
channel-by-channel basis).  Although this appears quite different, we
show that it is related in the sense that our lower bounds on the
number of \cnot{} gates for circuit topologies that are able to
generate all quantum channels of Kraus rank $K$ are also lower bounds
for almost all (in a mathematical sense) quantum channels of Kraus
rank $K$ individually, where the Kraus rank of a channel is defined as
the smallest number of Kraus operators required to represent the
channel and is equal to the rank of the corresponding Choi
state~\cite{choi}. 

It is worth emphasizing that there is a (measure zero) set of channels
for which our lower bounds do not apply individually, and this set
contains experimentally interesting channels. In other words, there are
circuits of lower cost than those given in this paper if the channel
has a simple or special form. Nevertheless, our constructions could
still be used as a starting point to find a low-cost circuit in such
cases. Further optimizations could then be performed with algorithms
such as, for example, the one given in~\cite{maslov}.

For certain special cases, the theory of decomposing operations is
quite developed. Considerable effort has been made to reduce the
number of \cnot{} gates required in the QCM for general unitary
gates~\cite{Knill, Aho, Vartiainen, unif_rot,10, 2} and state
preparation~\cite{10, 3}, which are both examples of a wider class of
operations, isometries. Recently, it was shown that every isometry
from $m$ to $n$ qubits can be implemented by using about twice the
\cnot{} count required by a lower bound~\cite{Iso}.  This 
leads to a method to implement quantum channels by using Stinespring's
theorem~\cite{Stinespring}, which states that every quantum channel
from $m$ to $n$ qubits can be implemented by an isometry from $m$ to
$m+2n$ qubits, followed by tracing out $m+n$ qubits. The
isometry can be decomposed into single-qubit gates and $4^{m+n}$
\cnots{} to leading order~\cite{Iso}. Working in the quantum circuit
model this \cnot{} count is optimal up to a factor of about four to
leading order~\cite{Iso}. However, one can significantly lower this
\cnot{} count and the required number of ancillas in more general
models.

Quantum operations beyond isometries have been investigated
in~\cite{one_ancilla}. Although~\cite{one_ancilla} did not focus
on a decomposition into elementary gates, combining the
decomposition in~\cite{one_ancilla} with an idea given
in~\cite{binary_search} and with the circuits for isometries given
in~\cite{Iso}, leads to low-cost decompositions of quantum channels
into single-qubit and \cnot{} gates in the MeasuredQCM using only one ancilla qubit. The combination
of~\cite{one_ancilla} and~\cite{binary_search} was fleshed
out in~\cite{Shen}, which appeared shortly after the first
version of the present work.  In~\cite{Shen}, several
applications are also discussed.

In this work, we give a new decomposition and proof that also leads to
near-optimal circuits for quantum channels.  In contrast to the work
mentioned above, we consider channels that map between spaces with
different dimensions. Our decomposition can be used for arbitrary
channels from $m$ to $n$ qubits (if $m=0$ our channels allow the preparation of arbitrary mixed states). In spite of the different proof
technique, the form of decomposition has similarities with the one
based on~\cite{one_ancilla}, which we discuss later.

Previously, the task of minimizing the number of required \cnot{}
gates for the implementation of quantum channels in the MeasuredQCM has been studied in
the case of channels on a single qubit~\cite{Wang_qubit}.  In the
special case of a single-qubit channel, we recover a circuit
topology (consisting of only one \cnot{} gate) similar to that given
in~\cite{Wang_qubit}. We also note that Ans\"atze for decompositions of
arbitrary channels have been considered
in~\cite{Wang_qudit,Wang_new}. One of these, Ansatz~1
in~\cite{Wang_new}, is based on applying the Cosine-Sine decomposition
to the Stinespring dilation isometry of the channel, and hence will
always work~\cite{Iso}. Our results imply that the Ansatz given
in~\cite{Wang_qudit} (which is designed for the RandomQCM) cannot work in
general because it does not have enough parameters.\footnote{Note that
  some of the phrasing in~\cite{Wang_qudit} gives the impression that
  this Ansatz is proven to work in all cases; however the authors
  confirmed that this is not intended.} Further Ans\"atze are given
in~\cite{Wang_new}, but it is not proven whether these work. In
contrast, our constructive decompositions are proven to always work.

\bigskip

In the following we describe how to construct circuit topologies for
quantum channels in the two aforementioned generalizations of the
quantum circuit model. Our asymptotic results are summarized in
Table~\ref{Table1}. First, we show that in the QCM with additional
classical randomness (for free) the number of required ancillas can be
reduced to $m$ and the \cnot{} count to $2^{2m+n}$ to leading
order. Moreover, we derive a lower bound in this model, which shows
that $m$ ancillas are necessary and that our \cnot{} count is optimal
up to a factor of about two to leading order.

Second, we show that the MeasuredQCM offers further improvement. Our
main result is a decomposition scheme for arbitrary $m$ to $n$
channels, which leads to the lowest known \cnot{} count of
$m\cdot2^{2m+1}+2^{m+n}$ if $m<n$ and of $n\cdot2^{2m+1}$ if
$m\geqslant n$ (to leading order). Moreover, our construction shows
that we can implement $m$ to $n\leqslant m$ channels using only $m+1$
qubits (i.e., one ancilla), and $m$ to $n>m$ channels using $n$ qubits
(which is clearly minimal, because the output of the channel is an
$n$-qubit state). 

Our construction also leads to
low-cost implementations of $m$ to $n$ channels for small $m$ and
$n$  (as does the construction resulting from the combination of~\cite{one_ancilla},~\cite{binary_search} and~\cite{Iso}). We give the explicit MeasuredQCM topologies for $m$ to $n$
channels for $1\leqslant m,n\leqslant 2$ in
Appendix~\ref{app:small_cases}. These circuits are most likely to be
of practical relevance for experiments performed in the near
future. In particular, they show that every one to two channel can be
implemented with 4 \cnot{} gates, every two to one channel with 7 and
every two to two channel with 13. These counts are lower than those
achieved by working in the QCM or the RandomQCM. For example, the best
known implementation of a two to two channel in the quantum circuit
model requires about 580 \cnots{}.\footnote{This is an an upper bound
  based on the Column-by-Column Decomposition for
  isometries~\cite{Iso}.} Allowing classical randomness reduces this
count to 54 \cnots{},\footnote{This corresponds to the \cnot{} count
  for a two to four isometry~\cite{Iso}.} which is over four times our
\cnot{} count of
13 when measurement and classical control are also allowed.\\

In future work, it would be interesting to generalize our circuit
constructions for other universal gate sets. This could be achieved by
finding circuits for isometries and then applying our construction for
channels in the MeasuredQCM described in Section~\ref{sec:UB_MCC},
which works independently of the chosen gate set. The ultimate goal
would be to design an algorithm that takes as input a given set of
gates, a noise model, an accuracy tolerance and a desired operation,
and that gives as output a circuit composed of gates from the set that would
approximate the desired operation to within the accuracy tolerance (if
this is possible), with the number of gates in the circuit being close
to minimal. Note that the constructions introduced in this paper could
be used as a subroutine in a version of this algorithm, and could
serve as a starting point to which further optimizations (to remove
gates where possible) are applied.

\begin{table}[!t] 
\renewcommand{\arraystretch}{1.4}
% if using array.sty, it might be a good idea to tweak the value of
% \extrarowheight as needed to properly center the text within the cells
\caption{Asymptotic upper and lower bounds on the number of \cnot{}
  gates for $m$ to $n$ channels in the three different models (Model~1: QCM,
  Model~2: RandomQCM, and Model~3: MeasuredQCM). The total number of qubits required for the constructions is also indicated.}
\label{Table1}
\centering
\begin{ruledtabular}
\begin{tabular}{llll}
Model&Lower bound&Upper bound&Qubits	\\ \hline 
1~\cite{Iso} &$\frac{1}{4} 4^{m+n}$&$4^{m+n}$&$m+2n$\\
2&$\frac{1}{2}2^{2m+n}$&$2^{2m+n}$&$m+n$\\
3 $(m<n)$&$\frac{1}{6}\left(2^{m+n+1}-2^{2m}\right)$&$m\cdot2^{2m+1}+2^{m+n}$&$n$\\
3 $(m\geqslant n)$&$\frac{1}{6}2^{2n}$&$n\cdot 2^{2m+1}$&$m+1$\\ [0.04cm]										   
\end{tabular}
\end{ruledtabular}
\end{table}

\section{Decomposition allowing classical randomness}
In the following, we consider the implementation of quantum channels in the RandomQCM, i.e., we allow classical randomness for free. Since the set of all quantum channels
from $m$ to $n$ qubits is convex, every $m$ to $n$ channel
$\mathcal{E}$ can be decomposed into a (finite) convex combination of
extreme $m$ to $n$ channels $\mathcal{E}_j$.\footnote{For a bound on the number of channels required see~\cite{MB}.}  

Physically this means that, allowing classical randomness, the channel
$\mathcal{E}=\sum_{j=1}^J p_j \mathcal{E}_j$ can be implemented by
performing the channel $\mathcal{E}_j$ with probability $p_j$ (and
forgetting about the outcome $j$).

\subsection{Upper bound}
By Remark~6 of~\cite{choi},\footnote{\label{fm_choi}In particular,
  Theorem~5 (and Remark~6) of~\cite{choi} characterizes the extreme
  points of the set of all completely positive, unital maps.  But the
  theorem (and the remark) can be adapted to trace preserving
  (completely positive) maps by considering the adjoint map
  (see~\cite{Friedland} or~\cite{MB} for more details).} every extreme
channel from $m$ to $n$ qubits has Kraus rank at most
$2^m$. Stinespring's theorem~\cite{Stinespring} then implies that in
order to implement every extreme channel it suffices to be able to
implement arbitrary isometries from $m$ to $m+n$ qubits.
Decompositions of such isometries use $2^{2m+n}$ \cnot{} gates to
leading order~\cite{Iso}. In the following section, we derive a lower
bound on the number of \cnot{} gates and ancilla qubits required for
$m$ to $n$ channels allowing classical randomness, which shows that
the \cnot{} count stated above is optimal up to a factor of two in
leading order and optimal in the number of required ancillas.

\subsection{Lower bound}\label{sec:CR_LB}
Because extreme channels cannot be decomposed into a convex
combination of other channels, classical randomness cannot help
implement them. Hence, a lower bound for extreme channels in the QCM
is also a lower bound for channels in the RandomQCM. Since the set of
extreme channels of Kraus rank $2^m$ is nonempty~\cite{Friedland}, at
least $m$ ancillas are required (using fewer ancillas, we could only
generate channels of smaller Kraus rank).\footnote{By a similar
  argument, one can see that the implementation of channels in the
  quantum circuit model requires $m+n$ qubits.} To find a lower bound
on the number of \cnot{} gates required for a quantum circuit topology
for $m$ to $n$ extreme channels, we can use a parameter counting
argument, similar to the argument used to derive a lower bound for
unitaries~\cite{unitary_lowerb1, unitary_lowerb2} or for channels in
the quantum circuit model~\cite{Iso}.

First, we count the number of (real) parameters required to describe
the set of all extreme channels.\footnote{A rigorous mathematical
  approach of the parameter counting (using the dimension of a
  differentiable manifold) can be found in~\cite{MB} and confirms the
  naive count described here. Note that a parameter count derived within the framework of semi-algebraic geometry was first given
  in~\cite{Friedland}.} Every quantum channel $\mathcal{E}$ from $m$
to $n$ qubits with Kraus rank $K$ can be represented by Kraus
operators
$A_i \in \textnormal{Mat}_{\mathbb{C}}\left(2^n \times 2^m\right)$
such that $\sum_{i=1}^{K} A_i^{\dagger}A_i=I$ and
$\mathcal{E}(X)=\sum_{i=1}^K A_i X A_i^{\dagger}$ [for all
$X \in \textnormal{Mat}_{\mathbb{C}}\left(2^m \times
  2^m\right)$]~\cite{choi}.
By Theorem 5 of~\cite{choi},$^{\ref{fm_choi}}$ a channel $\mathcal{E}$
is extreme if and only if all elements of the set
$\{A_i^{\dagger}A_j \}_{i,j \in \{1,2,\dots,K\}}$ are linearly
independent. Each $m$ to $n$ channel $\mathcal{E}$ of Kraus rank
$K=2^m$ can be described by $K$ $2^n \times 2^m$ (complex) matrices
$A_i$, which satisfy $4^{m}$ independent (note that the matrix
$\sum_{i=1}^{K} A_i^{\dagger}A_i$ is Hermitian) conditions (over
$\mathbb{R}$). However, the Kraus representation is not unique. Two
sets of Kraus operators $\{A_i\}_{i \in \{1,2,\dots,K\}}$ and
$\{B_i\}_{i \in \{1,2,\dots,K\}}$ describe the same channel if and
only if there exists a unitary $U \in U(2^m)$, such that
$A_i=\sum_{j=1}^K (U)_{i,j}B_j$~\cite{choi}. Since a $2^m \times 2^m$
unitary matrix is described by $4^m$ parameters, we conclude that the
set of all extreme channels form $m$ to $n$ qubits is described by
$2^{2m+n+1}-2^{2m+1}$ parameters. Note that the condition that the elements in
$\{A_i^{\dagger}A_j \}_{i,j \in \{1,2,\dots,K\}}$ must be linearly
independent is an open condition for $K=2^m$ and can therefore be
ignored for the parameter counting.

A quantum circuit topology for extreme $m$ to $n$ channels must
therefore introduce at least $2^{2m+n+1}-2^{2m+1}$ parameters. Since
\cnot{} gates cannot introduce parameters into a circuit topology, all
the parameters have to be introduced by the single-qubit gates. We
work with the following single-qubit rotation
gates\begin{eqnarray} \label{eq4}
  R_{x}(\theta)&=&\left(\begin{array}{cc} \cos [ \theta/2 ] &- \I \sin
      [\theta/2] \\ - \I \sin [\theta/2] & \cos
                                           [\theta/2]  \end{array}\right);\\
\label{eq5}
	R_{y}(\theta)&=&\left(\begin{array}{cc} \cos [\theta/2] & -\sin [\theta/2] \\   \sin [\theta/2] & \cos [\theta/2]  \end{array}\right);\\
\label{eq6}
	R_{z}(\theta)&=&\left(\begin{array}{cc} e^{-\I \theta/2 } & 0
            \\   0 & e^{\I \theta/2} \end{array}\right)\, .
\end{eqnarray}

For every unitary operation $U \in U(2)$ acting on a single qubit, there exist real numbers $\alpha,\beta,\gamma$ and $\delta$ such that

\begin{equation} \label{eq_single_qubit_unitary} 
	U=e^{\I\alpha}R_z(\beta)R_y(\gamma)R_z(\delta).
\end{equation}

A proof of this decomposition can be found in~\cite{Buch}. Note that (by symmetry) equation~\eqref{eq_single_qubit_unitary} holds for any two orthogonal rotation axes. The statement above can be represented as a circuit equivalence as follows.
\[
\Qcircuit @C=1.0em @R=.46em {
  &\gate{U}&\qw&=&&\gate{R_z}&\gate{R_y}&\gate{R_z}&\qw }
\]
The wire represents a qubit and the time flows from left to right. We
ignore the global phase shift, because it is physically undetectable.

Let us consider $l$ qubits, $l-m$ of which start in a fixed (not necessarily product) state. We can act with a single-qubit gate on
each qubit at the beginning of the quantum circuit topology (introducing
$3l$ parameters). To introduce further (independent) parameters, we
have to introduce \cnot{} gates. Naively, one would expect that every
\cnot{} gate can introduce six new parameters by introducing a
single-qubit gate after the control- and one after the action-part of
it. But by the following commutation relation:
\[
\Qcircuit @C=0.6em @R=.46em {
&\gate{R_z}&\gate{R_y}&\gate{R_z} & \ctrl{2} &\qw&&& &\gate{R_z}&\gate{R_y}& \ctrl{2} &\gate{R_z}&\qw   \\
& &&&&&&	 =&& \\
&\gate{R_x}&\gate{R_y}&\gate{R_x} &\targ &\qw& &&  &\gate{R_x} &\gate{R_y}&\targ&\gate{R_x}&\qw	 \\
}
\]

each \cnot{} gate can introduce at most four  parameters. Since we
trace out $l-n$ qubits at the end of the circuit, the single-qubit
gates on these qubits can not introduce any parameters into the
circuit topology [which removes $3(l-n)$ parameters]. We conclude
that by using $r$ \cnot{} gates we can introduce at most $4r+3n$
parameters into the circuit topology. By the parameter count above, we
require $4r+3n\geqslant 2^{2m+n+1}-2^{2m+1}$ or equivalently
$r\geqslant 2^{2m-1}\left(2^n-1\right)-\frac{3}{4}n$ for a quantum
circuit topology that is able to perform arbitrary extreme channels
from $m$ to $n$ qubits.

\begin{rmk} [Lower bound for nonexact circuits]
The derived lower bound can be strengthened and made more general
(see~\cite{MB}):  
the set of all quantum circuit topologies that
have fewer than $\lceil 2^{2m-1}\left(2^n-1\right)-\frac{3}{4}n
\rceil$ \cnot{} gates, together,\footnote{By combinatorial arguments,
  there are only finitely many different quantum circuit topologies
  consisting of a fixed number of \cnot{} gates (w.l.o.g.\ we can
  consider circuit topologies in which we perform single-qubit gates
  on all qubits at the start of the circuit and two after each \cnot{}
  gate).} are not able to approximate every $m$ to $n$ extreme channel
arbitrarily well. In fact, they can only generate a closed set of
measure zero\footnote{Nevertheless, many interesting operations lie in
  this set. This is similar to the case of isometries, where, for example, the
  operation required to implement Shor's algorithm~\cite{shor} lies in
  (the analog of) this set~\cite{Iso}.} in the smooth manifold of
$m$ to $n$ extreme channels of Kraus rank $2^m$. Therefore, the lower
bound holds for almost every $m$ to $n$ extreme channel of Kraus rank
$2^m$ individually. 
\end{rmk}

\section{Decomposition allowing measurement and classical control} 
We now move to the consideration of quantum circuit topologies in the
MeasuredQCM where we allow measurements (of single qubits in the
computational basis) and classical control on the measurement results
(and an arbitrary number of ancillas).  This generalizes the model
used above, since classical randomness can be generated by preparing
ancilla qubits in a certain state (by acting with single-qubit
unitaries on them), performing measurements and then controlling the
parameters of a circuit topology on the measurement results. In the
following section we describe how to construct circuit topologies for
arbitrary $m$ to $n$ channels of Kraus rank $K$. Applying this
to extreme channels (which have Kraus rank at most $2^m$) leads
to the \cnot{} counts given in Table~\ref{Table1}. A similar result could be found by using the decomposition described in~\cite{one_ancilla} using binary search~\cite{binary_search}.

\subsection{Upper bound} \label{sec:UB_MCC}
Let $\mathcal{E}$ be a channel from $m$ to $n$ qubits with Kraus rank $K=2^k$ and Kraus operators $\{A_i\}_{i \in \{1,2,\dots,K\}}$, $A_i \in \textnormal{Mat}_{\mathbb{C}}\left(2^n \times 2^m\right)$. We define the matrix $V=[A_1;A_2;\dots;A_K] \in \textnormal{Mat}_{\mathbb{C}}\left(2^{n+k}\times 2^m\right)$, by stacking the Kraus operators.\footnote{For example, we have $[A_1;A_2]:=\left(\begin{array}{c}A_1 \\ A_2\end{array}\right)$.} Since $V^{\dagger}V=\sum_{i=1}^{K} A_i^{\dagger}A_i=I$, we can consider the matrix $V$ as an isometry from $m$ to $n+k$ qubits (which corresponds to a Stinespring dilation of the channel $\mathcal{E}$).
If $n+k=m$ or $k=0$, we can perform $\mathcal{E}$ by implementing $V$
and tracing out $k$ qubits afterwards. In all other
cases,\footnote{Note that for all channels from $m$ to $n$ qubits of
  Kraus rank $K=2^k$ we have that $n+k\geqslant m$ (cf. Lemma 6
  of~\cite{Friedland}).} we consider each half of the matrix $V$
separately and define $B_0=[A_1;A_2;\dots;A_{K/2}]$ and
$B_1=[A_{K/2+1};A_{K/2+2};\dots,A_{K}]$. By the QR-Decomposition (for
rectangular matrices), we can find unitary matrices $Q_0,Q_1 \in
U(2^{n+k-1})$ and $R_0,R_1 \in \{[T;0;\dots;0] \in
\textnormal{Mat}_{\mathbb{C}}\left({2^{n+k-1} \times 2^m}\right): T
\in \textnormal{Mat}_{\mathbb{C}}\left({2^m \times 2^m}\right)$ is
upper triangular$\}$, such that $Q_0R_0=B_0$ and $Q_1R_1=B_1$. Note
that $Q_0$ and $Q_1$ are not unique: indeed only the first $2^m$
columns are determined and the others are free (up to
orthonormality). We can therefore consider $Q_0$ and $Q_1$ as
isometries from $m$ to $n+k-1$ qubits. Summarized, we have $(Q_0\oplus
Q_1)[R_0;R_1]=V$ and hence, $R:=[R_0;R_1]=(Q_0\oplus Q_1)^{\dagger}V$
is an isometry. We can represent this decomposition as an equivalence
of circuit topologies on $n+k$ qubits, where the first $n+k-m$ start
in the state $\ket{0}$ via
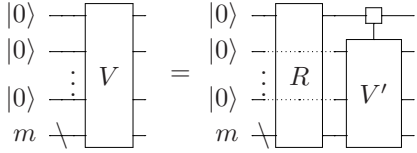
\[
\Qcircuit @C=0.45em @R=.46em {			
\lstick{\ket{0}}&\qw&\qw& \multigate{4}{V}&\qw	&&&&&&&&&			\lstick{\ket{0}}&\qw& \multigate{4}{R}&\qw	&\gate{}  \qwx[1] &\qw  \\
\lstick{\ket{0}}&\qw&\qw& \ghost{V} &\qw	&&&&&&&&&			\lstick{\ket{0}}&\qw&	\qwd&	\qwd& \multigate{3}{V'} &\qw  \\
&&\vdots&	&			&&&=&&&&&&&	\vdots	&&&&			\\
\lstick{\ket{0}}&\qw&\qw& \ghost{V}&\qw	&&&&	&&&&&			\lstick{\ket{0}}&\qw& \qwd &\qwd	&\ghost{V'} &\qw \\
\lstick{m}& {\backslash} \qw	&\qw&\ghost{V} &\qw	&&&&&&&&&	\lstick{m}			&{\backslash}\qw  & \ghost{R}	&\qw& \ghost{V'}&\qw 
}
\]
where the backslash stands for a data bus of several (in this case $m$) qubits and $V'=\{Q_0,Q_1\}$ is a placeholder for two isometries in  $\textnormal{Mat}_{\mathbb{C}}\left({2^{n+k-1} \times 2^m}\right)$. The unfilled square denotes a uniform control.\footnote{The notation of ``uniform control'' was introduced in~\cite{unif_rot}. Some authors also call these gates ``multiplexed'' (for example, see~\cite{2}).} In the case above, we implement $Q_0$ if the most significant qubit is in the state $\ket{0}$ and $Q_1$ if it is in the state $\ket{1}$. Note that the gate $R$ only acts nontrivially on the most significant and the $m$ least significant qubits. In particular, the second to $(n+k-m)$th qubits are still in the state $\ket{0}$ after applying $R$  (the lack of action of a gate on a particular qubit is indicated by use of a dotted line for that qubit).  We can apply the same procedure to the isometries $Q_0$ and $Q_1$. We repeat this $\tilde{k}$ times, until we end up with a quantum circuit topology of the form
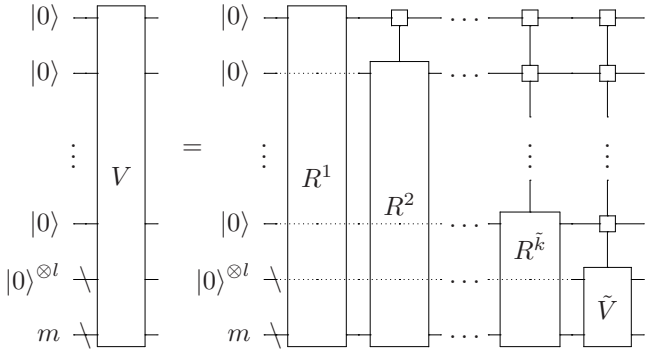
\[
\Qcircuit @C=0.45em @R=1.2em {			
\lstick{\ket{0}}&\qw& \multigate{7}{V}&\qw	&&&&&&&&&			\lstick{\ket{0}}&\qw& \multigate{7}{R^1}&\qw	&\gate{}  \qwx[1] &\qw&&\hdots&&&\gate{}  \qwx[1] &\qw&\gate{}  \qwx[1]&\qw \\
\lstick{\ket{0}}&\qw& \ghost{V}&\qw	&&&&&&&&&			\lstick{\ket{0}}&\qw& \qwd&\qwd	&  \multigate{6}{R^2}  &\qw &&\hdots&&&\gate{}  \qwx[1]&\qw&\gate{}  \qwx[1]&\qw \\
&&	&			&&&&&&&&&&		&&&&&	&	&&&&	&		\\
\vdots&&	&			&&&=&&&&&&\vdots&&		&&&&&		&&&	\vdots&&	\vdots	\\
&&	&			&&&&&&&&&&		&&&&&	&&	&&&&	&		\\
\lstick{\ket{0}}&\qw& \ghost{V}&\qw	&&&&	&&&&&			\lstick{\ket{0}}&\qw&  \qwd&\qwd& \qwd &\qwd&&\hdots&&& \multigateUp{2}{R^{\tilde{k}}}{1} \qwx[-1]&\qw&\gate{}  \qwx[1] \qwx[-1]&\qw \\
\lstick{\ket{0}^{\otimes l}}&{\backslash}\qw& \ghost{V} &\qw	&&&&&&&&&			\lstick{\ket{0}^{\otimes l}}	&{\backslash}\qw& \qwd	&\qwd	& \qwd	&\qwd&&\hdots&&& \qwd	&\qwd&\multigate{1} {\tilde{V}}&\qw  \\
\lstick{m}&{\backslash}\qw&\ghost{V} &\qw	&&&&&&&&&			\lstick{m}	&{\backslash}\qw& \ghost{R^1}&\qw	& \ghost{R^2}&\qw&&\hdots&&& \ghost{R^{\tilde{k}}}&\qw& \ghost{\tilde{V}}&\qw  \\
}
\]
where $\tilde{V} \in \textnormal{Mat}_{\mathbb{C}}\left({2^{m+l}
    \times 2^m}\right)$ is an isometry and where each gate $R^i$ acts nontrivially only on the  $i$th and the $m$ least significant qubits. If $m<n$, we set $l=n-m$ and $\tilde{k}=k$ and if
$m\geqslant n$, we set $l=1$ and $\tilde{k}=n+k-m-1$. Recall that we
can implement the channel $\mathcal{E}$ by applying the isometry $V$
and tracing out the first $k$ qubits afterwards (which we can think
about as performing measurements on them and forgetting the
result). Since measurements commute with controls, we conclude that
the following MeasuredQCM topology is able to perform all channels from $m$ to $n$ qubits of Kraus rank at most $K$
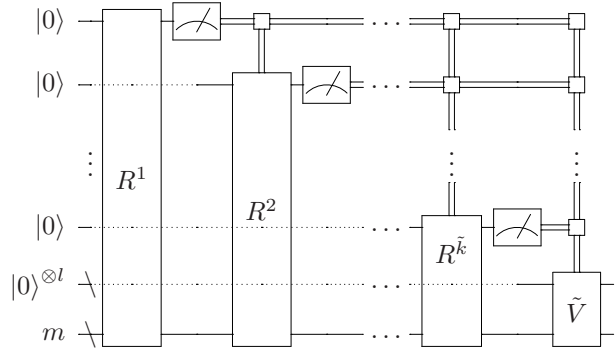
\[
\Qcircuit @C=0.45em @R=1em {
\lstick{\ket{0}}&\qw& \multigate{7}{R^1}&\meter	&\ucc{}  \cwx[1] &\cw&\cw&&\hdots&&&\ucc{}  \cwx[1] &\cw&\ucc{}  \cwx[1]&\\
\lstick{\ket{0}}&\qw& \qwd&\qwd	& \multigate{6}{R^2}&\meter &\cw &&\hdots&&&\ucc{}  \cwx[1]&\cw&\ucc{}  \cwx[1]& \\
&&&&		&&&	&	&&&&&&&&&&	\\
&\vdots&&&		&&&	&	&&&\vdots&&\vdots&&&&&	\\
&&&&		&&&	&	&&&&&&&&&&	\\
\lstick{\ket{0}}&\qw& \qwd&\qwd	&\qwd&\qwd&\qwd&&\hdots&&& \multigateUp{2}{R^{\tilde{k}}}{1} \cwx[-1]&\meter&\ucc{}  \cwx[1]  \cwx[-1]& \\
\lstick{\ket{0}^{\otimes l}}	&{\backslash}\qw& \qwd&\qwd	&\qwd&\qwd&\qwd&&\hdots&&& \qwd&\qwd&\multigate{1}{\tilde{V}}  &\qw \\
\lstick{m}	&{\backslash}\qw& \ghost{R^1}&\qw	& \ghost{R^2}&\qw&\qw&&\hdots&&& \ghost{R^{\tilde{k}}}&\qw& \ghost{\tilde{V}}  &\qw \\
}
\]
where we also measure the first $k-\tilde{k}$ of the $l+m$ least
significant qubits. Note that the circuit above can be implemented
with only one ancilla qubit by resetting it to the state $\ket{0}$
after the measurements and saving the measurement outputs in classical
registers

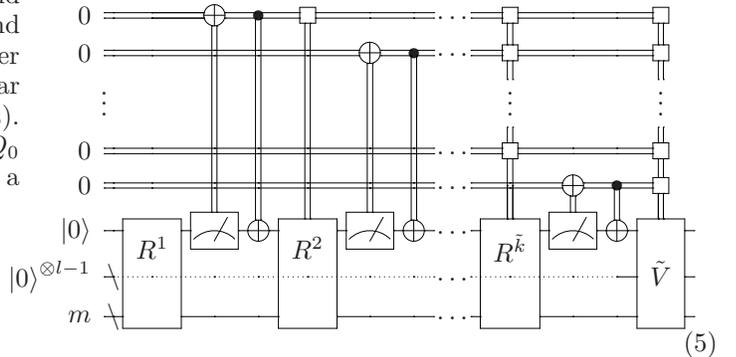
\begin{equation}\label{circ}
\Qcircuit @C=0.35em @R=0.6em {
&&&&\lstick{0}&\cw& \cw&\ctarg	 \cw \cwx[8]&\control	 \cw \cwx[8]&\ucc{}  \cwx[8]	&\cw	&\cw&\cw&&\hdots&&&\ucc{}  \cwx[1] &\cw&\cw&\ucc{}  \cwx[1]&\\
&&&&\lstick{0}&\cw&\cw&\cw& \cw&\cw &\ctarg&\control	 \cw \cwx[7]&\cw &&\hdots&&&\ucc{}  \cwx[1]&\cw&\cw&\ucc{}  \cwx[1]& \\
&&&&&&&&&&		&&&		&&&&&&&&&&&&	\\
&&&&\vdots&&&&&&		&&	&&&&&\vdots&&&\vdots&&&&&	\\
&&&&&&&&&&		&&&		&&&&&&&&&&&&	\\
&&&&&&&&&&		&&&		&&&&&&&&&&&&	\\
&&&&\lstick{0}&\cw&\cw& \cw	&\cw &\cw&\cw&\cw&\cw &&\hdots&&&\ucc{}  \cwx[2] \cwx[-1]&\cw&\cw&\ucc{} \cwx[-1] \cwx[1]& \\
&&&&\lstick{0}&\cw&\cw& \cw	&\cw &\cw&\cw&\cw&\cw &&\hdots&&&\cw&\ctarg&\control\cw\cwx[1]&\ucc{}  \cwx[1]& \\
&&&&\lstick{\ket{0}}&\qw& \multigateUp{2}{R^1}{1}&\meter
	&\targ& \multigateUp{2}{R^2}{1} &\meter
\cwx[-7]&\targ&\qw
&&\hdots&&& \multigateUp{2}{R^{\tilde{k}}}{1} \cwx[-1]&\meter \cwx[-1]&\targ&\multigate{2}{\tilde{V}}    \cwx[-1]&\qw  \\
&&&&\lstick{\ket{0}^{\otimes
    {l-1}}}&{\backslash}\qw&\qwd&\qwd&\qwd&\qwd&\qwd&\qwd&\qwd
&&\hdots&&&\qwd&\qwd&\qwd&\ghost{\tilde{V}}&\qw \\
&&&&\lstick{m}	&{\backslash}\qw& \ghost{R^1}&\qw&\qw	&
\ghost{R^2}&\qw&\qw&\qw
&&\hdots&&& \ghost{R^{\tilde{k}}}&\qw&\qw& \ghost{\tilde{V}}  &\qw \\
}
\end{equation}
where the second symbol means that a \nt{} is performed on the first
classical register if the output of the first measurement is one.

The construction above can be implemented on a system consisting of
$l+m$ qubits. The number of \cnot{} gates $N(m,n,k)$ required for the
MeasuredQCM topology above is $\tilde{k}N_{\textnormal{Iso}}(m,m+1)+N_{\textnormal{Iso}}(m,m+l)$. Working out
the different cases, we conclude that the number of \cnot{} gates
required for a quantum channel from $m$ to $n$ qubits of Kraus rank
$2^k$ is $N(m,n,k)=N_{\textnormal{Iso}}(m,n)$ if $k=0$,
$N(m,n,k)=N_{\textnormal{Iso}}(m,m)$ if $n+k=m$, and otherwise
\begin{subnumcases}{N(m,n,k)\leqslant} 
kN_{\textnormal{Iso}}(m,m+1)+N_{\textnormal{Iso}}(m,n) \text{ if } m<n \nonumber\\
(k+n-m)N_{\textnormal{Iso}}(m,m+1)    \text{ if } m \geqslant n ,\nonumber
\end{subnumcases}
where $N_{\textnormal{Iso}}(m,n)$ denotes the number of \cnot{} gates
required for an $m$ to $n$ isometry. If $n$ is large, we have
$N_{\textnormal{Iso}}(m,n) \simeq 2^{m+n}$ (for a more precise count,
see~\cite{Iso}). Note that the gates $R^i$ are isometries of a special
form, which could in principle be implemented by using fewer \cnots{}
than an arbitrary isometry. For simplicity, we have not accounted for
this in our \cnot{} counts.  The structure of the gates $R^i$ could be
significant when comparing our decomposition to that
of~\cite{one_ancilla}, which has a similar form to~\eqref{circ} but
where the isometries $R^i$ are general (rather than upper
triangular).\footnote{On a technical level, the reason for the lack of
  structure corresponds to the use of the polar decomposition rather
  than the QR-decomposition.}

Note that the main idea behind our construction and the requirement of
at most one ancilla is general: any decomposition scheme for
isometries (including with other universal gate sets; see,
e.g.,~\cite{OtherUniversalGateSet}) can be applied to
$R^1,R^2,\dots,R^{\tilde{k}}$ and $\tilde{V}$ arising in the
decomposition.

\subsection{Lower bound}
We expect that allowing measurement and classical controls cannot help when implementing isometries. Since isometries are special cases of channels, we expect further that a MeasuredQCM topology for $m$ to $n$ channels requires $\Omega \left(2^{m+n} \right)$ \cnot{} gates if $m<n$ and  $\Omega \left(4^{n} \right)$ \cnot{} gates if $m>n$~\cite{Iso}. Since the proof of this fact is quite technical and uses similar arguments as used to derive the lower bound for extreme channels above, we defer it to Appendix~\ref{app:LB}.  The result is summarized in Table~\ref{Table1}. Note that the lower bound for the case where $m>n$ is quite weak and it would be interesting to improve it in future work.

\section{Acknowledgements}
We acknowledge financial support from the European Research Council (ERC Grant Agreement no 337603), the Danish Council for Independent Research (Sapere Aude),  the Innovation Fund Denmark via the Qubiz project and VILLUM FONDEN via the QMATH Centre of Excellence (Grant No. 10059). R.C.\ acknowledges support from the EPSRC's Quantum Communications Hub. The authors are grateful to Bryan Eastin and Steven T. Flammia,
whose package {\sf Qcircuit.tex} was used to produce the circuit diagrams.
We thank Jonathan Home for helpful discussions.

\appendix

\section{Circuits for $m$ to $n$ channels for $1\leqslant m,n \leqslant2$} \label{app:small_cases}

The decomposition scheme in the MeasuredQCM described in Section~\ref{sec:UB_MCC} also leads to low-cost circuits for extreme $m$ to $n$
channels for small $m$ and $n$. In the following, we demonstrate how
to find circuits for $m$ to $n$ channels in the cases where
$1\leqslant m,n\leqslant 2$.  

\noindent\emph{1 to 1 channels}|An extreme channel from one to one
qubit (which is of Kraus rank at most two) can be implemented by
performing a one to two isometry followed by tracing out the first
qubit.  We can use the circuit topology for one to two isometries from
Appendix~B1 of~\cite{Iso}:
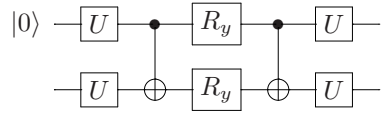
\[
\Qcircuit @C=1.0em @R=.46em {
&\lstick{\ket{0}}&\gate{U}&\ctrl{2} &\gate{R_y}&\ctrl{2} &\gate{U}&\qw\\
                          &&             &&\\
&& \gate{U}&\targ&\gate{R_y} &\targ&\gate{U}&\qw  \\
}
\]
Noting that a unitary before a partial trace can be removed, and that
controls commute with measurements, we obtain the following circuit
for a one to one channel:
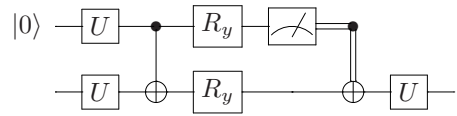
\[
\Qcircuit @C=1.0em @R=.46em {
&\lstick{\ket{0}}&\gate{U}&\ctrl{2} &\gate{R_y}&\meter &\control \cwx[2]\cw&&\\
                     &&             &&&\\
&& \gate{U}&\targ&\gate{R_y}&\qw &\targ&\gate{U}&\qw   \\
}
\]
Therefore, any single-qubit channel can be implemented with one
\cnot{} gate. A similar circuit topology was first derived
in~\cite{Wang_qubit}. \bigskip

\noindent\emph{1 to 2 channels}|We do the decomposition exactly as described in the general case in~Section~\ref{sec:UB_MCC}. This leads to a circuit topology of the form 
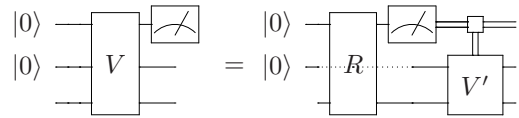
\[
\Qcircuit @C=0.45em @R=.46em {			
\lstick{\ket{0}}&\qw&\qw& \multigate{2}{V}&\meter 	&&&&&&&&&			\lstick{\ket{0}}&\qw& \multigate{2}{R}&\meter &\ucc{}\cwx[1]\cw&\cw  \\
\lstick{\ket{0}}&\qw&\qw& \ghost{V} &\qw	&&&=&&&&&&			\lstick{\ket{0}}&\qw&	\qwd&	\qwd& \multigate{1}{V'} &\qw  \\
& \qw	&\qw&\ghost{V} &\qw	&&&&&&&&&			& & \ghost{R}	&\qw& \ghost{V'}&\qw 
}
\]
where $V$ is a $1$ to $3$ isometry corresponding to a Stinespring dilation of the implemented channel and $R$ and $V'$ denote $1$ to $2$ isometries. We use the circuit topology for one to two isometries given in Appendix~B1 of~\cite{Iso} (consisting of two \cnot{} gates). Therefore, an extreme channel from one to two qubits (of Kraus rank at most two) requires $2\cdot N_{\textnormal{Iso}}(1,2)=4$ \cnot{} gates.\bigskip

\noindent\emph{2 to 1 channels}| A channel from two qubits to one qubit of Kraus rank at most four can be implemented by an isometry from two to three qubits and tracing out the first two qubits afterwards. We do the first few steps of the decomposition of an two to three isometry as in Appendix~B2b of~\cite{Iso}. This leads to the circuit topology 
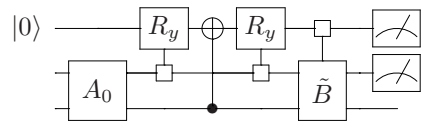
\[
\Qcircuit @C=0.5em @R=.2em {
\lstick{\ket{0}}	& \qw   &\gate{R_{y}}  \qw             &\targ	 &\gate{R_{y}}  \qw	&\gate{} \qwx[1] \qw &   \qw &   \meter  	\\
&\multigate{1}{A_0} &\gate{} \qwx[-1] \qw	&  \qw &\gate{} \qwx[-1] \qw & \multigate{1}{\tilde{B}}	 &  \qw&   \meter     \\
 	&\ghost{A_0}		&  \qw  &\ctrl{-2}			&	  \qw &\ghost{\tilde{B}} &	  \qw	&	  \qw&	   \\
}
\]
where $A_0$ and $\tilde{B}$ are two qubit unitaries. We can use a technical trick introduced in Appendix B of~\cite{2} to save one \cnot{} gate: by Theorem~14 of~\cite{2}, we can decompose the gate $A_0$ into a part  (which we denote by $\hat{A_0}$) consisting of two \cnot{} gates (and single-qubit gates) and a diagonal gate $\Lambda$. 
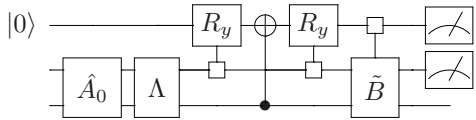
\[
\Qcircuit @C=0.5em @R=.2em {
\lstick{\ket{0}}	& \qw & \qw   &\gate{R_{y}}  \qw             &\targ	 &\gate{R_{y}}  \qw	&\gate{} \qwx[1] \qw &   \qw &   \meter 	\\
&\multigate{1}{\hat{A_0}}	&\multigate{1}{\Lambda} &\gate{} \qwx[-1] \qw	&  \qw &\gate{} \qwx[-1] \qw & \multigate{1}{\tilde{B}}	 &  \qw&   \meter  \\
 	&\ghost{\hat{A_0}}&\ghost{\Lambda	}&  \qw  &\ctrl{-2}			&	  \qw &\ghost{\tilde{B}} &	  \qw	&	  \qw\\
}
\]
Note that we reversed the gate order of the circuit given in Theorem~14 of~\cite{2}, such that the diagonal gate is performed after the gate $\hat{A_0}$. We commute the diagonal gate $\Lambda$ to the right and merge it with the gate $\tilde{B}$ (and call the merged gate $\hat{B}$). Therefore, and since controls commute with measurements, the circuit topology given above is equivalent to the following.
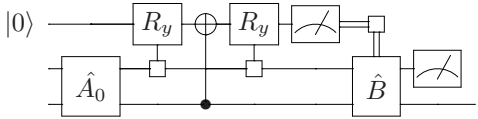
\[
\Qcircuit @C=0.5em @R=.2em {
\lstick{\ket{0}}	& \qw   &\gate{R_{y}}  \qw             &\targ	 &\gate{R_{y}}  \qw &   \meter	&\ucc{} \cwx[1] &  \\
&\multigate{1}{\hat{A_0}} &\gate{} \qwx[-1] \qw	&  \qw &\gate{} \qwx[-1] \qw &  \qw& \multigate{1}{\hat{B}}	 &   \meter  \\
 	&\ghost{\hat{A_0}}		&  \qw  &\ctrl{-2}			&	  \qw &	  \qw &\ghost{\hat{B}} &	  \qw	&	  \qw \\
}
\]
We decompose the uniformly controlled $R_y$ gates as described in Theorem~8 of~\cite{2}. Noting that 2 \cnot{} gates cancel out each other, we get the following circuit topology.
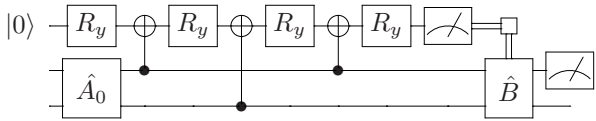
\[
\Qcircuit @C=0.5em @R=.2em {
\lstick{\ket{0}} 		 &\gate{R_{y}}  			&\targ 	&\gate{R_{y}}           &\targ	 	&\gate{R_{y}}   &\targ 	&\gate{R_{y}}	 &   \meter			&\ucc{} \cwx[1] &  	\\
				&\multigate{1}{\hat{A_0}}		&\ctrl{-1} 	& \qw 			&\qw 	&  \qw		&\ctrl{-1}	&  \qw  		&  \qw			& \multigate{1}{\hat{B}}	 &   \meter   \\
 				&\ghost{\hat{A_0}}			&  \qw 	&  \qw  			&\ctrl{-2}	&  \qw  		&  \qw 	&  \qw    		&  \qw			&\ghost{\hat{B}} &	  \qw		 \\
}
\]
We can further save one \cnot{} gate in the decomposition of the gate $\hat{B}$. By~\cite{unitary_lowerb1, unitary_lowerb2}, we have the following equivalence of circuit topologies.
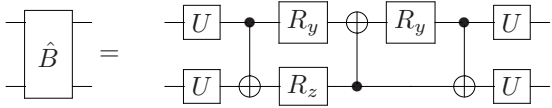
\[
\Qcircuit @C=0.7em @R=.46em {
& \multigate{2}{\hat{B}}&\qw&&&&  &\gate{U}&\ctrl{2} &\gate{R_y}&\targ&\gate{R_y}&\ctrl{2}&\gate{U}&\qw\\
&  & &          =                                   &&&\\
&\ghost{\hat{B}}  &\qw&&  &&&  \gate{U}&\targ&\gate{R_z} &\ctrl{-2}&\qw&\targ&\gate{U}&\qw \\
}
\]
Since controls commute with measurements, we get the following equivalence.
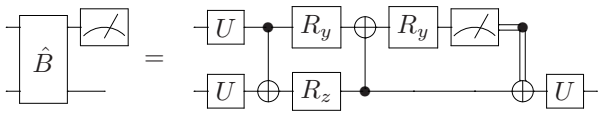
\[
\Qcircuit @C=0.5em @R=.46em {
& \multigate{2}{\hat{B}}&\meter&&&&&&\gate{U}&\ctrl{2} &\gate{R_y}&\targ&\gate{R_y}&\meter &\control \cwx[2] \cw &&\\
&  &      &&   =                            &&             &&&\\
&\ghost{\hat{B}} &\qw &&&&&&  \gate{U}&\targ&\gate{R_z} &\ctrl{-2}&\qw&\qw&\targ&\gate{U} &\qw  \\
}
\]
Substituting this circuit into the second-to-last one, we find a circuit topology  for channels from two qubits to one qubit of Kraus rank at most four (and hence, in particular, for extreme two to one channels) consisting of $7$ \cnot{} gates.\bigskip

\noindent\emph{2 to 2 channels}|This case works similarly to the
case of two to one channels of Kraus rank at most four. We use the
CSD-approach (cf.~\cite{Iso}) to decompose the isometries arising from
the decomposition scheme described in~Section~\ref{sec:UB_MCC}, and apply
the technical tricks introduced in the Appendix of~\cite{2}. Indeed, decomposing the first two to three isometry arising in the decomposition described  in~Section~\ref{sec:UB_MCC} as described above in the case of two to one channels, we find the following circuit for (extreme) two to two channels of Kraus rank at most four
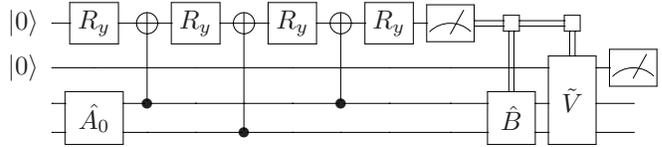
\[
\Qcircuit @C=0.5em @R=.2em {
\lstick{\ket{0}} 		 &\gate{R_{y}}  			&\targ 	&\gate{R_{y}}           &\targ	 	&\gate{R_{y}}   &\targ 	&\gate{R_{y}}	 &   \meter			&\ucc{} \cwx[2] &\ucc{} \cwx[1] &  	\\
\lstick{\ket{0}} 		 &\qw 			&\qw  	&\qw           &\qw 	 	&\qw    &\qw  	&\qw 	 &\qw 			&\qw & \multigate{2}{\tilde{V}}	 & \meter	\\
				&\multigate{1}{\hat{A_0}}		&\ctrl{-2} 	& \qw 			&\qw 	&  \qw		&\ctrl{-2}	&  \qw  		&  \qw			& \multigate{1}{\hat{B}}	 &\ghost{\tilde{V}} &  \qw   \\
 				&\ghost{\hat{A_0}}			&  \qw 	&  \qw  			&\ctrl{-3}	&  \qw  		&  \qw 	&  \qw    		&  \qw			&\ghost{\hat{B}} &	\ghost{\tilde{V}} &	 \qw	 \\
}
\]
where $\tilde{V}$ denotes the second two to three isometry arising in the decomposition described  in~Section~\ref{sec:UB_MCC}. We can merge the gate $\hat{B}$ into $\tilde{V}$, which leads to
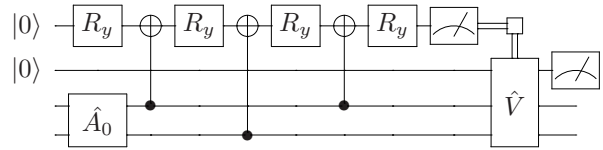
\[
\Qcircuit @C=0.5em @R=.2em {
\lstick{\ket{0}} 		 &\gate{R_{y}}  			&\targ 	&\gate{R_{y}}           &\targ	 	&\gate{R_{y}}   &\targ 	&\gate{R_{y}}	 &   \meter			&\ucc{} \cwx[1] &  	\\
\lstick{\ket{0}} 		 &\qw 			&\qw  	&\qw           &\qw 	 	&\qw    &\qw  	&\qw 	 &\qw 			 & \multigate{2}{\hat{V}}	 & \meter	\\
				&\multigate{1}{\hat{A_0}}		&\ctrl{-2} 	& \qw 			&\qw 	&  \qw		&\ctrl{-2}	&  \qw  		&  \qw				 &\ghost{\hat{V}} &  \qw   \\
 				&\ghost{\hat{A_0}}			&  \qw 	&  \qw  			&\ctrl{-3}	&  \qw  		&  \qw 	&  \qw    		&  \qw			 &	\ghost{\hat{V}} &	 \qw	 \\
}
\]
where $\hat{V}$ is a two to three isometry. Therefore, we can again apply the decomposition scheme described above for two to one channels to $\hat{V}$. Since we do not measure the third qubit at the end of the circuit, we use 8 \cnot{} gates to decompose the gate $\hat{V}$. We conclude that we can decompose any channel from two to two qubits of Kraus rank at most four  (and hence, in particular, any extreme two to two channel) with at most $13$ \cnot{} gates.

\section{Lower bound for isometries in the MeasuredQCM} \label{app:LB}

We give a lower bound on the number of \cnot{} gates required for a
MeasuredQCM topology that is able to generate all
isometries from $m$ to $n$ qubits using the basic gate library
comprising arbitrary single-qubit unitaries and \cnot. A lower bound
for $m$ to $n$ isometries in the quantum circuit model was already
given in~\cite{Iso}. However, here we work in a more general model
than that of~\cite{Iso}, since we allow measurements and classical
controls (and an arbitrary number of ancillas, each of which start in
the state $\ket{0}$).

Let us consider an arbitrary MeasuredQCM topology for $m$ to $n$
isometries consisting of $p \geqslant n$ qubits. The most general
sequence of operations that can be performed by such a circuit
topology, is as follows. We perform a certain gate sequence on the $p$
qubits, before the first qubit is measured. Then we perform a second
gate sequence on $p-1$ qubits, which may be controlled on the
measurement result of the first qubit. Then we measure the second
qubit. In the case where we want to measure the first and second qubit
together, the second gate sequence can be chosen to be trivial. We go
on like this until we have measured $p-n$ qubits.  We then forget
about the measurement results at the end of the MeasuredQCM
topology. Note that the reuse of a qubit after a measurement can be
incorporated into the above procedure by adding an additional ancilla
qubit and copying the measurement outcome there. We conclude that any
MeasuredQCM topology for $m$ to $n$ isometries consisting of
$p \geqslant n$ qubits can be represented in the following form
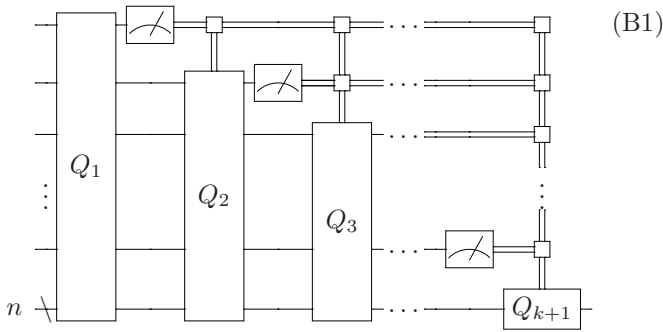
\begin{equation} \label{eq:general_form}
\Qcircuit @C=0.4em @R=.8em {
&\qw& \multigate{7}{Q_1}&\meter&\ucc{}  \cwx[1] &\cw &\ucc{} \cwx[1] &\cw&&\hdots&&&\cw&\cw&\ucc{} \cwx[1]&  \\
&\qw& \ghost{Q_1}&\qw& \multigate{6}{Q_2} &\meter &\ucc{} \cw \cwx[1] &\cw&&\hdots&&&\cw&\cw&\ucc{} \cwx[1] &\\
&\qw& \ghost{Q_1}&\qw&\ghost{Q_2}  &\qw & \multigate{5}{Q_3}&\qw&&\hdots&&&\cw&\cw& \ucc{} \cwx[1] &\\
&&&&&&&&&&&&&&& \\
&\vdots&&&&&&&&&&&&&\vdots& \\
&&&&&&&&&&&&&&&\\
&\qw&\ghost{Q_1} &\qw& \ghost{Q_2}&\qw&\ghost{Q_3}&\qw&&\hdots&&&\qw&\meter&\ucc{} \cwx[1] \cwx[-1]&  \\
\lstick{n}&{\backslash}\qw& \ghost{Q_1} &\qw& \ghost{Q_2}&\qw& \ghost{Q_3}&\qw&&\hdots& &&\qw&\qw&\gate{Q_{k+1}}&\qw
}
\end{equation}
where $k:=p-n$ and we can think of $Q_i$ as the set of $(p+1-i)$-qubit
unitaries that can be generated by the corresponding quantum circuit
topology.  In other words, there is a first quantum circuit topology (perhaps with
free parameters), followed by a measurement, then a classically
controlled quantum circuit topology conditioned on the outcome, followed by a
second measurement and so on.

\begin{thm} [Lower bound in the MeasuredQCM] \label{thm:lower:bound_MCC} A
  MeasuredQCM topology that is able to generate all
  isometries from $m$ to $n\geqslant m$ qubits using ancillas
  initialized in the state $\ket{0}$ has to consist of at least
  $\lceil\frac{1}{6}\left(2^{n+m+1}-2^{2m}-\textnormal{max}(2, 3m)-1\right)\rceil$
  \cnot{} gates.
  \end{thm}

\begin{rmk}
  The lower bound given in Theorem~\ref{thm:lower:bound_MCC} is by a
  constant factor of $\frac{2}{3}$ (to leading order) lower than the
  one for isometries in the quantum circuit model of
  $\lceil \frac{1}{4}\left(2^{n+m+1}-2^{2m}-2n-m-1\right) \rceil$
  \cnot{} gates~\cite{Iso}. Intuitively, the use of ancillas,
  measurements and classical controls should not be helpful for
  implementing isometries. Therefore, we expect that the lower bound
  given in Theorem~\ref{thm:lower:bound_MCC} can be improved.
\end{rmk}

Since isometries from $m$ to $n$ qubits are special cases of $m$ to $n \geqslant m$ channels, we get the following Corollary.
\begin{cor}
A MeasuredQCM topology that is able to generate all channels from $m$ to $n\geqslant m$ qubits has to consist of at least $\lceil\frac{1}{6}\left(2^{n+m+1}-2^{2m}-\textnormal{max}(2, 3m)-1\right)\rceil$ \cnot{} gates.
\end{cor}
Moreover, we find the following lower bound for  $m$ to $n <m$ channels.
\begin{cor}
A MeasuredQCM topology that is able to generate all channels from $m$ to $n<m$ qubits has to consist of at least $\lceil\frac{1}{6}\left(4^{n}-3n-1\right)\rceil$ \cnot{} gates.
\end{cor}
\begin{proof}
Assume to the contrary that there exists a MeasuredQCM topology consisting of fewer than $\lceil\frac{1}{6}\left(4^{n}-3n-1\right)\rceil$ \cnot{} gates that is able to generate all channels from $m$ to $n<m$ qubits. Such a topology must, in particular, be able to implement all $n$-qubit unitaries from the first $n$ input qubits to the $n$ output qubits (independently of the state of the other $m-n$ input qubits). We can turn this topology into a MeasuredQCM topology for unitaries on $n$ qubits by fixing the state of the last $m-n$ input qubits to $\ket{0}$. But such a topology cannot exist by Theorem~\ref{thm:lower:bound_MCC}.
\end{proof}

Before giving the proof of Theorem~\ref{thm:lower:bound_MCC}, we
sketch the idea. We start with a circuit topology of the
form~\eqref{eq:general_form} consisting of $p \geqslant n$ qubits,
$p-m$ of which are initially in the state $\ket{0}$, and assume that
it is able to generate all isometries from $m$ to $n$ qubits. In
principle, one would expect that a circuit topology controlled on one
(randomized) classical bit can introduce twice as many parameters as
the circuit topology itself, and hence that controlling on measurement
results can help to reduce the \cnot{} count (as we saw in
Section~\ref{sec:UB_MCC}). However, in the special case where we want
to implement isometries, the classical control cannot increase the
number of introduced parameters. The reason for this is related to the
fact that the distribution of the measurement outputs are independent
of the input state of the isometry.  The precise statement is given in
the following Lemma.

\begin{lem} [Independence of measurement results] \label{lem:1} Assume that the
  whole circuit in~\eqref{eq:general_form} performs an isometry from
   $m$ to  $n$ qubits for a certain choice of the free parameters of the MeasuredQCM topology. Then the distribution of the measurement outcomes is independent of the input state of the isometry.
\end{lem}
\begin{proof}
  It suffices to show this for all non-orthogonal states.\footnote{If
    all non-orthogonal states have the same distribution, then all
    states do, since the distribution for two orthogonal states
    $\ket{\psi_0}$ and $\ket{\psi_1}$ must then agree with that of any
    third state $\ket{\psi}$ that is not orthogonal with both.} Take
  two non-orthogonal input states $\ket{\psi_0}$ and $\ket{\psi_1}$
  and assume to the contrary that there exists a measurement $M$
  in~\eqref{eq:general_form}, whose output distribution is different
  depending on which of these states is input. Let $P$ be the
  distribution over the outcomes for $M$ if we choose the input state
  $\ket{\psi_0}$, and $Q$ be the analogous probability distribution if
  we choose the input state $\ket{\psi_1}$. Since we are implementing
  an isometry, the output states $\ket{\psi'_0}:=V \ket{\psi_0}$ and
  $\ket{\psi'_1}:=V \ket{\psi_1}$ can be turned back into
  $\ket{\psi_0}$ and $\ket{\psi_1}$.  If we repeat this procedure $t$
  times, then, the distribution of outcomes for $M$ is either the
  i.i.d.\ distribution $P^{\times t}$ or the i.i.d.\ distribution
  $Q^{\times t}$. Since $Q\neq P$ by assumption, these two
  distributions can be distinguished arbitrarily well for large enough
  $t$. This contradicts the fact that in any measurement procedure the
  maximum probability of correctly guessing which of these states is
  given as an input is
  $\frac{1}{2}\left[1+D(\ketbra{\psi_0}{\psi_0},\ketbra{\psi_1}{\psi_1})\right]<1$,
  where $D$ is the trace distance. \end{proof}

To handle the independence of the measurement distributions on the
input state, it is useful to introduce the concept of postselection
(see also~\cite{Aaronson}). We introduce the Postselected Quantum
Circuit Model (PostQCM for short) as a modification of the QCM to
include also single-qubit projectors onto the states $\ket{0}$ and
$\ket{1}$ at the end of the circuit.\footnote{Note that this is
  equivalent to a measurement in the $\{\ket{0},\ket{1}\}$ basis and
  postselecting on one of the outcomes.} Note that the single-qubit
projectors correspond to linear maps that are not unitary.  We say
that a PostQCM topology with associated total linear map $C$
implements the isometry
$V$, if $C=c V$, where $c\neq0$ is some complex number.\\

We say that a PostQCM topology corresponds to a MeasuredQCM topology
of the form~\eqref{eq:general_form}, if it can be obtained
from~\eqref{eq:general_form} using the following procedure.  First,
every measurement is replaced by a single-qubit projector (onto
either $\ket{0}$ or $\ket{1}$). Then all classical controls are
removed. Finally we move the single-qubit projectors to the end of the
circuit. Note that the number of single-qubit gates and \cnots{} of a
circuit topology of the form~\eqref{eq:general_form} is the same as
that of the PostQCM topology formed by making these replacements.

\begin{lem} \label{lem:PostQCT} The set of isometries that can be
  generated by a MeasuredQCM topology of the
  form~\eqref{eq:general_form} is a subset of the set of isometries
  that can be generated by all the corresponding PostQCM topologies
  together.
\end{lem}
\begin{proof}
  Assume that an isometry $V$ can be generated by a MeasuredQCM
  topology of the form~\eqref{eq:general_form} for a certain choice of
  its free parameters. Hence, by Lemma~\ref{lem:1}, the distribution
  of the measurement outputs is independent of the input state of the
  isometry. Therefore, the circuit must perform the isometry
  regardless of the measurement outputs and hence we can choose and
  fix an arbitrary output which occurs with nonzero probability. In
  other words, we can replace each measurement
  in~\eqref{eq:general_form} with a single-qubit projector onto
  $\ket{0}$ if the probability of measuring $0$ is nonzero, and with
  a single-qubit projector onto $\ket{1}$ otherwise. Note that this
  circuit can still perform the isometry $V$. Removing the classical
  controls, which does not change the action performed by the whole
  circuit, we obtain a corresponding PostQCM topology that is able to
  generate $V$.
\end{proof}

\begin{lem} \label{lem:PostQCT_LB_parameters} A PostQCM topology that
  has fewer than $2^{n+m+1}-2^{2m}-1$ free parameters can only
  generate a set of measure zero of the set of all $m$ to $n$
  isometries (where we identify isometries that only differ by a
  global phase).
\end{lem}

\begin{proof}
The argument works similarly to the arguments used in~\cite{MB,
  unitary_lowerb1, unitary_lowerb2}. Let us denote by $C$ the linear
map corresponding to the PostQCM topology. We can think of this map as sending a certain choice of $d$ real parameters $(\theta_1,\dots,\theta_d)$ of the PostQCM topology to the corresponding $2^n\times 2^m$ matrix $C(\theta_1,\dots,\theta_d)$, which describes the whole action of the circuit. We restrict the domain of the free parameters to the set $D \subset \mathbb{R}^d$, such that for all  $(\theta_1,\dots,\theta_d) \in D$ there exists an isometry $V$ and a complex number $c\neq 0$, such that $C(\theta_1,\dots,\theta_d)=cV$. We denote the set of one dimensional unitaries by $U(1)$ and define the orbit space $V_{m,n}/U(1)$, which corresponds to the set of all $m$ to $n$ isometries, after quotienting out the (physically undetectable) global phase. We denote the corresponding (smooth) quotient map by $\pi:V_{m,n}\mapsto V_{m,n}/U(1)$ (see~\cite{MB} for more details). Then, we define the smooth map $T(\theta_1,\dots,\theta_d):=\pi \circ \tfrac{C(\theta_1,\dots,\theta_d)}{\sqrt{2^{-m} \tr{\,C(\theta_1,\dots,\theta_d)^{\dagger}C(\theta_1,\dots,\theta_d)}}}:D \mapsto V_{m,n}/U(1)$. By Sard's theorem, $T(D)$ is of measure zero in $V_{m,n}/U(1)$ if $d < \dim(V_{m,n}/U(1))=2^{m+n+1}-2^{2m}-1$.
\end{proof} 

\begin{lem}  \label{lem:PostQCT_LB} 
A PostQCM topology that consists of fewer than $\lceil\frac{1}{6}\left(2^{n+m+1}-2^{2m}-\textnormal{max}(2, 3m)-1\right)\rceil$ \cnot{} gates (and an arbitrary number of ancilla qubits initialized in the state $\ket{0}$) can only generate a set of measure zero of the set of all $m$ to $n$ isometries.
\end{lem}
\begin{proof}
  We may assume $n>1$ (for $n=1$ the statement of the Lemma is
  trivial). By Lemma~\ref{lem:PostQCT_LB_parameters}, we have left to
  show that a PostQCM topology consisting of fewer than
  $\lceil\frac{1}{6}\left(2^{n+m+1}-2^{2m}-\textnormal{max}(2, 3m)-1\right)\rceil$ \cnots{}
  cannot introduce $2^{m+n+1}-2^{2m}-1$ or more (independent) real
  parameters. Since single-qubit projections do not introduce
  parameters into the circuit, all parameters must be introduced by
  single-qubit gates. To relate the number of single-qubit rotations
  to the number of \cnot{} gates, we use similar arguments to those
  used in Section~\ref{sec:CR_LB} to derive the lower bound for
  channels allowing classical randomness. We again use the commutation
  properties of \cnot{} gates and single-qubit rotations, which show
  that a \cnot{} can introduce at most four parameters. However, in
  contrast to Section~\ref{sec:CR_LB}, we commute all single-qubit
  rotations to the left (instead of to the right) and use the fact
  that the first single-qubit rotation on an ancilla can introduce at
  most two parameters (because an ancilla qubit always starts in the
  state $\ket{0}$ and two parameters are enough to describe an
  arbitrary single-qubit pure state). Note that, in general, the
  single-qubit rotations performed directly before a single-qubit
  projection have a nontrivial effect on the operation performed by
  the whole circuit. Thus, a PostQCM topology with $q$ \cnot{} gates
  and consisting of $p\geqslant n$ qubits can introduce at most
  $4q+2(p-m)+3m$ parameters.  Note that we may assume
  $q\geqslant \textnormal{min}(p-m, p-1)$, since otherwise, there exists a collection of
  ancilla qubits and output qubits (which are not input qubits) that
  are not quantum-connected to the $m$ input qubits.\footnote{This
    follows from a simple statement in graph theory, that a connected
    graph must have at least $V-1$ edges, where $V$ denotes the number
    of vertices of the graph.} Any unconnected output qubits that are
  not input qubits start in the state $\ket{0}$ and always remain
  product with the other output qubits.\footnote{If all output qubits
    are not quantum connected to the input qubits, the PostQCM
    topology can generate only a fixed output state independent on the
    input state, and hence is not able to perform any isometry.} For
  $n>1$, the set of isometries for which the output state always has a
  product form has fewer parameters than the set of arbitrary
  isometries, and is hence of measure zero. In the case that all the
  unconnected qubits are ancilla qubits, they have a trivial effect on
  the performed circuit and can be removed without affecting the
  action of the circuit. Therefore, a PostQCM topology with $q$
  \cnot{} gates can introduce at most $6q+\textnormal{max}(2, 3m)$ parameters and hence, a
  circuit topology consisting of fewer than
  $\lceil\frac{1}{6}\left(2^{n+m+1}-2^{2m}-\textnormal{max}(2, 3m)-1\right)\rceil$ \cnots{}
  cannot introduce $2^{m+n+1}-2^{2m}-1$ (or more) parameters.
\end{proof}

\begin{proof}[Proof of Theorem~\ref{thm:lower:bound_MCC}]
  Consider a MeasuredQCM topology of the form~\eqref{eq:general_form}
  consisting of fewer than
  $\lceil \frac{1}{6}\left(2^{n+m+1}-2^{2m}-\textnormal{max}(2, 3m)-1\right)\rceil$ \cnot{}
  gates. Since each of the corresponding PostQCM topologies consists
  of the same number of \cnot{} gates, each can only generate a set of
  measure zero in the set of all $m$ to $n$ isometries by
  Lemma~\ref{lem:PostQCT_LB}. Since the MeasuredQCM
  topology~\eqref{eq:general_form} has at most $2^k$ corresponding
  PostQCM topologies, the set of isometries that can be generated by
  all corresponding PostQCM topologies together is still of measure
  zero. The theorem then follows from Lemma~\ref{lem:PostQCT} and the
  fact that a subset of a set of measure zero is again of measure
  zero.
\end{proof}

\end{document}